\documentclass{article}
\usepackage{fullpage,latexsym,amssymb,stmaryrd,theorem,fleqn}
\usepackage{graphicx}
\usepackage{color}

\newtheorem{all}{Proposition}

\newtheorem{proposition}[all]{Proposition}

\newtheorem{corollary}[all]{Corollary}

{\theorembodyfont{\upshape}}
{\theorembodyfont{\upshape}}
\newenvironment{proof}{\noindent
{\bf Proof:}}{$\dashv$}
\newenvironment{proof*}{\noindent
{\bf Proof:}}{$\dashv$}

\newcommand{\oomit}[1]{}
\font\hollow=msbm10
\newcommand{\Rset}{{\mbox{\hollow R}}}
\newcommand{\Iset}{{\mbox{\hollow I}}}

\newcommand{\Var}{\mathrm{Var}}
\newcommand{\VarA}{\mathrm{Var}_{:=}}

\newcommand{\interp}{\mathit{A}}
\newcommand{\jnterp}{\mathit{B}}
\newcommand{\knterp}{\mathit{C}}
\newcommand{\chop}{{}^{\frown}}

\newcommand{\sem}[1]{\left[\!\left[#1\right]\!\right]}

\def\fin{\mathrm{fin}}
\def\const{\mathrm{const}}
\def\loc{\mathrm{loc}}

\def\dom{\mathrm{dom}}
\newcommand{\rarrow}[1]{\stackrel{#1}{\longrightarrow}}
\def\terminate{\checkmark}

\def\extchoice{[]}
\newcommand{\BoxI}[1][0mu]{\mathord{
    \mbox{\rlap{$\mskip 4.9mu\mskip#1$\relax
        \raisebox{1.5pt}{$\scriptscriptstyle\mathrm{i}$}}}
    \Box}}

\newcommand{\BoxS}[1][0mu]{\mathord{
    \mbox{\rlap{$\mskip 4.9mu\mskip#1$\relax
        \raisebox{1.5pt}{$\scriptscriptstyle\mathrm{s}$}}}
    \Box}}

\newcommand{\Boxdot}{\BoxI^\circ}
\newcommand{\BoxSdot}{\BoxS^\circ}
\def\Esplit{\exists_{\|}}
\def\Asplit{\forall_{\|}}

\def\rvec#1{\overrightarrow{#1}}

\begin{document}

\title{Compositional Hoare-style Reasoning about Hybrid CSP\\
 in the Duration Calculus}
\author{Dimitar Guelev\thanks{Institute of Mathematics and Informatics, Bulgarian Academy of Sciences, e-mail:{\tt gelevdp@math.bas.bg}}, Shuling Wang and Naijun Zhan\thanks{State Key Laboratory of Computer Science, Institute of Software, Chinese Academy of Sciences, e-mail:{\tt \{wangsl,znj\}@ios.ac.cn}.}}

\maketitle

\begin{abstract}
Deductive methods for the verification of hybrid systems vary on the format of statements in correctness proofs. Building on the example of Hoare triple-based reasoning, we have investigated several such methods for systems described in Hybrid CSP, each based on a different assertion language, notation for time, and notation for proofs, and each having its pros and cons with respect to  expressive power, compositionality and practical convenience. In this paper we propose a new approach based on weakly monotonic time as the semantics for interleaving, the Duration Calculus (DC) with infinite intervals and general fixpoints as the logic language, and a new meaning for Hoare-like triples which unifies assertions and temporal conditions. We include a proof system for reasoning about the properties of systems written in the new form of triples that is complete relative to validity in DC.
\end{abstract}

\section{Introduction}

{\em Hybrid systems} exhibit combinations of discrete and
continuous evolution, the typical example being a continuous plant with discrete control. A number of abstract models and requirement specification languages have been proposed for the verification of hybrid systems, the commonest model being {\em hybrid automata}~\cite{DBLP:conf/hybrid/AlurCHH92,DBLP:conf/hybrid/MannaP92,DBLP:conf/lics/Henzinger96}. Hybrid CSP (HCSP)~\cite{Jifeng:1994:CHS:197600.197614,DBLP:conf/hybrid/ChaochenJR95} is a process algebra which extends CSP by constructs for continuous evolution described in terms of ordinary differential equations, with domain boundary- and communication-triggered interruptions. The mechanism of synchronization is message passing. Because of its compositionality, HCSP can be used to handle complex and open systems. Here follows an example of a simple generic HCSP description of a continuously evolving plant with discrete control:
\[
 \begin{array}{ll}
 &(\mathbf{while}\ \top\ \mathbf{do}\ \langle F(\dot x,x,u)=0\rangle\unrhd \mathit{sensor}!x\rightarrow\mathit{actuator}?u)\ \|
\\
&(\mathbf{while}\ \top\ \mathbf{do}\ (\mathbf{wait}\  d;\mathit{sensor}?s;\mathit{actuator}!C(s)))
 \end{array} \]
The plant evolves according to some continuous law $F$ that depends on a control parameter $u$. The controller samples the state of the plant and updates the control parameter once every $d$ time units.

In this paper we propose a Hoare-style proof system for reasoning about hybrid systems which are modelled in HCSP. The features of HCSP which are handled by the logic include communication, timing constraints, interrupts and continuous evolution governed by differential equations. Our proof system is based on the Duration Calculus (DC, \cite{ZHR91,ZH04}), which is a first-order real-time temporal logic and therefore enables the verification of HCSP systems for temporal properties. DC is an interval-based temporal logic. The form of the satisfaction relation in DC is
$I,\sigma\models\varphi$,
where $\varphi$ is a temporal formula, $I$ is an interpretation of the respective vocabulary over time, and $\sigma$ is a {\em reference interval} of real time, unlike {\em point-based} TLs, where a reference time point is used. The advantages of intervals stem from the  possibility to accommodate a complete execution of a process
and have reference to termination time points of processes as well as the starting points. Pioneering work on interval-based reasoning includes Allen's interval algebra and Halpern and Shoham's logic \cite{Allen:1983:MKT:182.358434,HS86}. ITLs have been studied in depth with respect to the various models of time by a number of authors, cf. e.g. \cite{DBLP:journals/jancl/GorankoMS04}. Since an interval can be described as the pair of its endpoints, interval-based logics are also
viewed as {\em two-dimensional} modal logics \cite{Ven91a,Ven91thesis}. Interval Temporal Logic (ITL) was first proposed and developed by Moszkowski \cite{Mos85,Mos86,CMZ} for discrete time, as a reasoning tool for digital circuits. DC can be
viewed as a theory in {\em real time} ITL. We use the infinite interval variant of DC which was proposed in \cite{ZDL95}, which allows
intervals whose right end is $\infty$
for modelling non-terminating behaviour. We include an operator for Kleene star in order to model iterative behaviour, to facilitate the handling of liveness properties. Axioms and proof rules about infinite time and Kleene star in DC can be found in \cite{WX04,GD04}.

Hoare-style proof systems are about proving {\em triples} of the form
$\{P\}\,\mathbf{code}\,\{Q\}$,
which stand for the {\em partial correctness} property $P(x)\wedge\mathbf{code}(x,x')\rightarrow Q(x)$. The meaning of triples generalizes to the setting of reactive systems in various ways, the common feature of them all being that $P$ and/or $Q$ are {\em temporal} properties. In our system $\mathbf{code}$ is a HCSP term, $P$ and $Q$ are written in DC. The intended meaning is
\begin{equation}\label{informalsemantics}
\begin{array}{p{4.3in}}
Given an infinite run which satisfies $P$ at some initial subinterval, $\mathbf{code}$ causes it to satisfy also $Q$ at the initial subinterval representing the execution of $\mathbf{code}$.
\end{array}
\end{equation}
The initial subinterval which is supposed to satisfy $P$, can as well be a degenerate ($0$-length) one. Then $P$ boils down to an assertion on the initial state. This interval can also be to the entire infinite run in question. In this case $P$ can describe conditions provided by the environment throughout runs. $Q$ is supposed to hold at an interval which exactly matches that of the execution of $\mathbf{code}$. In case $\mathbf{code}$ does not terminate, this would be the whole infinite run too. Using our DC semantics $\sem{.}$ for HCSP terms, the validity of $\{P\}\,\mathbf{code}\,\{Q\}$ is defined as the validity of
\[P\chop\top\Rightarrow\neg(\sem{\,\mathbf{code}\,}\wedge\neg (Q\chop\top))\]
at infinite intervals, which is equivalent to (\ref{informalsemantics}).

We exploit the form of triples to obtain a {\em compositional} proof system, with each rule corresponding to a basic HCSP construct. This forces proofs to follow the structure of the given HCSP term. Triples in this form facilitate assume-guarantee reasoning too. For instance,
\[\frac{\{A\}\,\mathbf{code}\,_1\{B\}\qquad\{B\}\,\mathbf{code}\,_2\{C\}}{\{A\}\,\mathbf{code}\,_1\|\,\mathbf{code}\,_2\{((B\chop\top)\wedge C)\vee(B\wedge(C\chop\top)\}}\]
where $\|$ denotes parallel composition, is an admissible rule in our proof system, despite not being among the basic rules. A detailed study of assume-guarantee reasoning about discrete-time reactive systems in terms of triples of a similar form with point-based temporal logic conditions can be found in \cite{DBLP:journals/toplas/AbadiL93}.

The main result about our proof system in the paper is its completeness relative to validity in DC.

\paragraph
{Structure of the paper} After brief preliminaries on HCSP, we propose a weakly-monotonic time semantics for it in terms of DC formulas and prove its equivalence to an appropriate operational semantics. Next we give our proof system and use the DC-based semantics to demonstrate its relative completeness. Finally we summarize a generalization of the approach where arbitrary fixpoints can be used instead of HCSP's tail recursion and the use of $\|$ and the respective rather involved proof rule can be eliminated. That turns out to require both the general fixpoint operator of DC \cite{DBLP:conf/csl/Pandya95} and the right-neighbourhood modality (cf. \cite{DBLP:series/eatcs/ChaochenH04}) to handle the meaning of $P$s in the presence of properly recursive calls. We conclude by discussing related work and make some remarks.
\section{Preliminaries}

\subsection{
Syntax and informal semantics of Hybrid CSP}
Process terms have the syntax

\vspace{2mm}
\noindent
{\small
$\begin{array}{lllp{2.4in}}
P, Q & ::= &
{\bf skip}\mid & do nothing;\\
%&& {\bf stop} \mid & terminate immediately;\\
&& x_1,\ldots, x_n :=e_1,\ldots, e_n  \mid & simultaneous assignment;\\
&& {\bf wait}\ d\mid {\bf await}\ b \mid &fixed time delay; wait until $b$ becomes true;\\
&& ch?x \mid ch!e \mid IO \mid  & input and output; communication-guarded choice;\\
&& \langle F(\dot x, x)=0\wedge b\rangle\mid & $x$ evolves according to $F$ as long as $b$ holds;\\
& &
\langle F(\dot x, x)=0\wedge b\rangle\unrhd IO\mid & evolve by $F$ until $\neg b$ or $IO$ becomes ready;\\
& & & terminate, if $\neg b$ is reached first;\\
& & & otherwise execute $IO$;\\
&& P;Q  \mid P\ \|\ Q \mid & sequential composition;
parallel composition\\
&& \mathbf{if}\ b\ \mathbf{then}\ P\ \mathbf{else}\ Q\mid P \sqcup Q \mid & conditional; internal non-deterministic choice;\\
& & \mu X.G & recursion.
\end{array}$
}
\vspace{2mm}

\noindent
In the above BNF, $IO$ has the following form:
\begin{equation}\label{IOsyntax}
%IO\rightleftharpoons
ch_1?x_1\rightarrow P_1\extchoice\ldots\extchoice ch_k?x_k\rightarrow P_k\extchoice ch_{k+1}!e_{k+1}\rightarrow P_{k+1}\extchoice\ldots\extchoice ch_n!e_n\rightarrow P_n
\end{equation}
for some arbitrary $k$, $n$, $x_1,\ldots,x_k$, $e_{k+1},\ldots,e_n$ and some distinct $ch_1,\ldots,ch_n$.
$IO$ engages in one of the indicated communications as soon as a partner process becomes ready, and then proceeds as the respective $P_i$. In $\mu X.G$, $G$ has the syntax
\[
\begin{array}{ll}\label{guardedBNF}
G::= &H\mid \rvec{x}:=\rvec{e};P\mid\langle F(\dot x,x)=0\wedge b\rangle;P\mid \mathbf{if}\ b\ \mathbf{then}\ G\ \mathbf{else}\ G \\
&\mid G\sqcup G\mid G;P\mid \mu Y.G\mid H\| H
\end{array}
\]
where $H$ stands for arbitrary $X$-free terms, $Y\not=X$ and $P$ can be any process term. This restricts $X$ of $\mu X.G$ to be {\em guarded} in $G$ and rules out occurrences of $X$ of $\mu X.G$ in the scope of $\|$ in $G$. The communication primitives $ch?x$ and $ch!e$ are not mentioned in the syntax for $G$ as they are treated as derived in this paper. They can be assigned the role of guards which $\rvec{x}:=\rvec{e}$ has in (\ref{guardedBNF}). Obviously $X$ is guarded in the $P_1,\ldots,P_n$ of $IO$ as in (\ref{IOsyntax}) too. Below we focus on the commonest instance of $\mu$
\begin{equation}\label{whiledef}
\mathbf{while}\ b\ \mathbf{do}\ P\rightleftharpoons \mu X.\mathbf{if}\ b\ \mathbf{then}\ (P;X)\ \mathbf{else}\ \mathbf{skip}
\end{equation}
CSP's Kleene star $P^*\rightleftharpoons \mu X.( \mathbf{skip}\sqcup (P;X))$,
which stands for some unspecified number of successive executions of $P$, is handled similarly. We explain how our setting ports to general fixpoints, and some technical benefits from that, other than the obvious gain in expressive power, in a dedicated section.

\subsection{The Duration Calculus}
We use DC with infinite intervals as in \cite{ZDL95,WX04} and Kleene star. The reader is referred to \cite{ZH04} for
a comprehensive introduction. Here we summarize only the main notions for the sake of self-containedness. DC is a classical first-order predicate modal logic with one normal binary modality called {\em chop} and written $\chop$. The time domain is $\Rset^\infty=\Rset\cup\{\infty\}$. Satisfaction has the form $I,\sigma\models\varphi$ where $I$ is an interpretation of the non-logical symbols, $\sigma\in\Iset(\Rset^\infty)$, $\Iset(\Rset^\infty)=\{[t_1,t_2];t_1\in\Rset,t_2\in\Rset^\infty,t_1\leq t_2\}$. {\em Flexible} non-logical symbols depend on the reference intervals for their meaning. {\em Chop} is defined by the clause
\[\begin{array}{lll}
I,\sigma\models(\varphi\chop \psi) &\mbox{iff}& \mbox{either there exists a }t\in\sigma\setminus\{\infty\}\mbox{ such that }I,[\min\sigma,t]\models\varphi \\
& & \mbox{and }I,[t,\max\sigma]\models\psi, \mbox{ or }\max\sigma=\infty\mbox{ and }I,\sigma\models\varphi.
\end{array}\]
Along with the usual first-order non-logical symbols, DC features boolean valued {\em state variables}, which form boolean combinations called {\em state expressions}. The value $I_t(S)$ of a state expression $S$ is supposed to change between $0$ and $1$ only finitely many times in every bounded interval of time.
{\em Duration terms} $\int \! \! S$ take a state expression $S$ as the operand and evaluate according to the clause
\[\textstyle I_\sigma(\int \!\! {S} )=\int\limits_{\min\sigma}^{\max\sigma}I(S)(t)dt.\]
$\ell$ is used for $\int ({\bf 0}\Rightarrow{\bf 0})$ and always evaluates to the length of the reference interval, and $\lceil S\rceil$ stands for $\ell\not=0\wedge\int \!\! S=\ell$ to denote that $S$ holds almost everywhere and the interval is non-degenerate. $\lceil S\rceil^0$ denotes just $\int \!\! S=\ell$.

In Section \ref{muSection} we use the converse neighbourhood modality $\Diamond_l^c$ and the least fixpoint operator $\mu$. The former appears in {\em Neighbourhood Logic} and the corresponding system of DC, and is defined by the clause
\[I,\sigma\models\Diamond_l^c\varphi\mbox{ iff }I,[\min\sigma,t]\models\varphi\mbox{ for some }t\in\Rset\cup\{\infty\},t\geq \min\sigma.\]
Formulas $\mu X.\varphi$, where $X$ is a dedicated type of variable, are well formed only if $\varphi$ has no negative occurrences of $X$. To define the meaning of $\mu X.\varphi$, $\varphi$ is regarded as the defining formula of a operator of type ${\mathcal P}(\Iset(\Rset^\infty))\rightarrow{\mathcal P}(\Iset(\Rset^\infty))$, $A\mapsto \{\sigma\in\Iset(\Rset^\infty):I_X^A,\sigma\models\varphi\}$, and $I,\sigma\models\mu X.\varphi$ iff $\sigma$ appears in the least fixpoint of this operator, which happens to be monotonic by virtue of the syntactic condition on $\varphi$. Kleene star is defined in terms of $\mu$ by the clause
\[\models\varphi^*\rightleftharpoons\mu X.(\ell=0\vee\varphi\chop X).\]

\section{Operational semantics of HCSP}
\label{opsem}

\paragraph{\bf Ownership of variables}
We write $\Var(P)$ for the set of the program variables which occur in $P$. Expressions of the form $\dot x$ in continuous evolution process terms are, syntactically, just program variables, and are restricted not to appear in arithmetical expressions $e$ outside the $F(\dot x,x)$ of continuous evolution terms, or on the left hand side of $:=$.
As it becomes clear below, the dependency between $x$ and $\dot x$ as functions of time is spelled out as part of the semantics of continuous evolution.
We write $\VarA(P)$ for all the variables in $\Var(P)$ which occur on the left hand side of $:=$, in $ch?$ statements, and the $x$-es or $\dot x$-es in any of the forms of continuous evolution within $P$. Parallel composition $P\| Q$ is well-formed only if $\VarA(P)\cap\VarA(Q)=\emptyset$.

\paragraph{\bf Modelling input and output}
We treat $ch?x$ and $ch!e$ as derived constructs as they can be defined in terms of dedicated shared variables $ch?$, $ch!$ and $ch$ after \cite{DBLP:conf/icalp/OlderogH83}:
\footnote{Hoare style proof rules for a system with $ch?x$ and $ch!e$ appearing as primitive constructs were proposed by Zhou Chaochen \emph{et al.} in \cite{GWZZ13}. That work features a different type of triples and follows the convention that process variables are not observable across threads thus ruling out a shared-variable emulation of communication.}
\oomit{\[{\small\begin{array}{lll}
ch!e&\rightleftharpoons& ch:=e;\\
& & ch!:= \top;\\
& & \mathbf{await}\ ch?;\\
& & \mathbf{await}\ \neg ch?;\\
& & ch!:= \bot
\end{array}
\qquad
\begin{array}{lll}
ch?x&\rightleftharpoons& ch?:=\top;\\
& & \mathbf{await}\ ch!;\\
& & x:=ch;\\
& & ch?:=\bot;\\
& & \mathbf{await}\ \neg ch!
\end{array}}\]
}
\[\begin{array}{lll}
ch!e&\rightleftharpoons& ch:=e;
ch!:= \top;
\mathbf{await}\ ch?;
\mathbf{await}\ \neg ch?;
 ch!:= \bot\\
ch?x&\rightleftharpoons& ch?:=\top;
\mathbf{await}\ ch!;
x:=ch;
ch?:=\bot;
\mathbf{await}\ \neg ch!
\end{array}\]
We assume that $ch!,ch\in \VarA(ch!e)$ and $ch?\in\VarA(ch?x)$.
Communication-guarded external choice $IO$ can be defined similarly. We omit the definition as it is lengthy and otherwise uninsightful. The other derived constructs are defined as follows:

\vspace{2mm}
\noindent
{\small $\begin{array}{lll}
\langle F(\dot x,x)=0\wedge b\rangle\unrhd_d Q &\rightleftharpoons& t:=0;\langle F(\dot x,x)=0\wedge\dot t=1\wedge b\wedge t\leq d\rangle;\\
&&\mathbf{if}\ \neg\overline{(\neg b)} \ \mathbf{then}\ Q
\ \mathbf{else}\ \mathbf{skip}\\[1mm]
\mathbf{wait}\ d &\rightleftharpoons& \langle 0=0\wedge\top\rangle\unrhd_d\mathbf{skip}\\[1mm]
\langle F(\dot x,x)=0\wedge b\rangle\unrhd IO &\rightleftharpoons&\langle F(\dot x,x)=0\wedge b\wedge\bigwedge\limits_{i\in I}\neg ch_i^*\rangle;\mathbf{if}\  \bigwedge\limits_{i\in I}\neg ch_i^*  \ \mathbf{then}\ \mathbf{skip}\ \mathbf{else}\ IO\end{array}$
}
\vspace{2mm}

\noindent
Here $ch_i^*$ stands for $ch_i?$, resp. $ch_i!$, depending on whether the respective action in $IO$ is input or output. To account of the impossibility to mechanically (and computationally) tell apart $x<c$ from $x\leq c$ about time-dependent quantities $x$, in $\neg\overline{(\neg b)}$ we use $\overline{a}$ for a condition that defines the topological closure of $\{\mathbf{x}:a(\mathbf{x})\}$. It is assumed that $\overline{b}$ admits a syntactical definition. E.g., $\overline{x<c}$ is $x\leq c$ for $x$ being a continuous function of time.

\paragraph{Reduction of HCSP process terms}
Next we define a {\em reduction relation} $P\rarrow{\interp,V} Q$ where $V$ is a set of process variables and
\begin{equation}\label{interpformat}
\interp:\sigma\rightarrow(V'\cup\{r,n\}\rightarrow\Rset^\infty\cup\{0,1\}),
\end{equation}
where $V'$ is a set of process variables, $r$ and $n$ are boolean variables outside $V'$ and $\sigma\in\Iset$.
In $\Rset^\infty\cup\{0,1\}$ we emphasize the additional use of $0,1\in\Rset$ as truth values.
We consider $P\rarrow{\interp,V} Q$ only for $V$ such that $\VarA(P)\subseteq V\subseteq V'$.
For HCSP terms $P$ in the scope of a $Q$ which on its turn is an operand of a $\|$, with no other $\|$s between this one and $P$, the semantics of $P$ must specify the behaviour of all the variables
from $\VarA(Q)$, which are {\em controlled} by the {\em enveloping thread} $Q$ of $P$. $V'\setminus V$ is meant to include of the variables which are not controlled by the enveloping thread of $P$ but still may be accessed in it. In the sequel we write $\dom\interp$ for $\sigma$ and $\Var(\interp)$ for $V'$ from (\ref{interpformat}).

If $V\subseteq\Var(\interp)$, then $\interp|_V$ stands for the restriction of $\interp$ to the variables from $V$. I.e., given $\interp$ as in (\ref{interpformat}),
\[\interp|_V: \sigma\rightarrow(V\cup\{r,n\}\rightarrow\Rset^\infty\cup\{0,1\}).\]
Given an arithmetic or boolean expression $e$ such that $V(e)\subseteq V(\interp)$, we write $\interp_t(e)$ for the value of $e$ under $\interp$ at time $t\in\dom \interp$. Given $\interp$ and $\jnterp$ such that $\max\dom\interp=\min\dom\jnterp$, $\Var(\interp)=\Var(\jnterp)$ and $\interp_{\max\dom \interp}(x)=\jnterp_{\min\dom \jnterp}(x)$ for all $x\in \Var(\interp)\cup\{r,n\}$, $\interp;\jnterp$ is determined by the conditions $\dom \interp;\jnterp=\dom \interp\cup \dom \jnterp$, $(\interp;\jnterp)_t(x)=\interp_t(x)$ for $t\in\dom\interp$ and $(\interp;\jnterp)_t(x)=\jnterp_t(x)$ for $t\in\dom\jnterp$ for all $x\in\Var(\interp)\cup\{r,n\}$. A complete and possibly infinite behaviour of $P$ can be defined as $\interp_1;\interp_2;\ldots$ where $P_{i-1}\rarrow{\interp_i,V}P_i$, and $P_0=P$.

\paragraph{The auxiliary variables $r$ and $n$}
To handle the causal ordering of computation steps without having to account of the {\em negligibly} small time delays they contribute, we allow stretches of time in which continuous evolution is 'frozen', and which are meant to just keep apart time points with different successive variable values that have been obtained by computation. Intervals of negligible time are marked by the boolean variable $r$. $P$ (or any of its descendant processes) claims exclusive control over the process variables during such intervals, thus achieving atomicity of assignment. Time used for computation steps by {\em any} process which runs in parallel with $P$ or $P$ itself is marked by $n$. Hence $\interp_t(r)\leq\interp_t(n)$, $t\in\dom\interp$ always holds in the mappings (\ref{interpformat}). As it becomes clear below, each operand $P_i$ of a $P_1\|P_2$  has its own $r$, and no two such variables evaluate to $1$ at the same time, which facilitates encoding the meaning of $\|$ by conjunction. In processes with loops and no other recursive calls, the $r$s can be enumerated globally. More involved form of recursive calls require the $r$s to be quantified away.

This approach is known as the {\em true synchrony hypothesis}. It was introduced to the setting of DC in \cite{PD97} and was developed in \cite{GD02,GD04} where properties $\varphi$ of the overall behaviour of a process in terms of the relevant continuously evolving quantities are written $(\varphi/\neg N)$, {\em the projection of $\varphi$ onto state $\neg N$}, which holds iff $\varphi$ holds at the interval obtained by gluing the $\neg N$-parts of the reference one. The approach is alternative to the use of {\em super-dense chop} \cite{HZ96chopping}.

\subsection{The reduction rules}

To abbreviate conditions on $A$ in the rules which generate the valid instances of $P\rarrow{\interp,V}Q$ below, given an $X\subseteq \Var(\interp)$ and a boolean or arithmetical expression $e$, we put:
\[\begin{array}{lll}
\const(X,\interp) &\rightleftharpoons & \bigwedge\limits_{x\in X}(\forall t\in \dom\interp)(\interp_t(x)=\interp_{\min\dom\interp}(x))\\
\const^\circ(X,\interp) &\rightleftharpoons & \bigwedge\limits_{x\in X}(\forall t\in \dom\interp\setminus\{\max\dom\interp\})(\interp_t(x)=\interp_{\min\dom\interp}(x))\\
\oomit{\const(e,\interp,a) &\rightleftharpoons & (\forall t\in \dom\interp)(\interp_t(e)=a)\\
}%oomit
\const^\circ(e,\interp,a)& \rightleftharpoons & (\forall t\in \dom\interp\setminus\{\max\dom\interp\})(\interp_t(e)=a)
\end{array}\]
In these abbreviations ${}^\circ$ indicates that no restriction is imposed at $\max\dom\interp$.

\paragraph{Reduction of process terms of the primitive types}
The valid transitions which are specific to primitive process terms are given by the following rules:

\noindent
$\begin{array}{lp{5in}}
\\
\displaystyle\frac{\max\dom\interp=\min\dom\interp}{\mathbf{skip}\rarrow{\interp,V}\terminate}
\qquad
\displaystyle\frac{\begin{array}{l}
\const(V\setminus\{x_1,\ldots,x_n\},\interp)\\
\const^\circ(\{x_1,\ldots,x_n\},\interp)\\
\const^\circ(r\wedge n,\interp,1)\\
\interp_{\max\dom\interp}(x_i)=\interp_{\min\dom\interp}(e_i),\ i=1,\ldots,n\\
\max\dom\interp<\infty
\end{array}
}
{x_1,\ldots,x_n:=e_1,\ldots,e_n\rarrow{\interp,V}\terminate}
\\
\\
\displaystyle\frac{\begin{array}{l}
\const^\circ(F(\dot x,x),A,0)\\
\const^\circ(\interp_t(\dot x)-\frac{d}{dt}\interp_t(x),\interp,0)\\
[2mm]
\const^\circ(\neg r\wedge\neg n\wedge b,\interp,1)\\
\const(V\setminus\{\dot x,x\},\interp)
\end{array}
}
{\langle F(\dot x,x)=0\wedge b\rangle\rarrow{\interp,V}\langle F(\dot x,x)=0\wedge b\rangle}
\qquad
\displaystyle\frac{\begin{array}{l}
\interp_{\min\dom\interp}(b)=0\\
\max\dom\interp=\min\dom\interp
\end{array}}
{\langle F(\dot x,x)=0\wedge b\rangle\rarrow{\interp,V}\terminate}
\\
\\
\end{array}$

\noindent
The rule about simultaneous assignment states that assignment takes an interval of negligible time, %znj
with its enveloping thread claiming exclusive control over the process variables. All variables except the ones which are assigned are kept constant throughout.

The first rule about $\langle F(\dot x,x)=0\wedge b\rangle$ describes 'ordinary' (non-negligible) periods in which continuous evolution takes place within the boundary condition $b$. The second rule describes states in which $b$ has become false and therefore evolution terminates immediately.

\paragraph{Reduction of compound terms}

\qquad

\noindent
$\begin{array}{l}\\
\displaystyle\frac{
P\rarrow{\interp,V}P'\quad P'\not=\terminate}
{P;Q\rarrow{\interp,V}P';Q
}
\qquad
\displaystyle\frac{
P\rarrow{\interp,V}\terminate\quad\max\dom\interp<\infty
}
{P;Q\rarrow{\interp,V}Q
}
\qquad
\displaystyle\frac{P\rarrow{\interp,V}P'\quad \max\dom\interp=\infty}
{P;Q\rarrow{\interp,V}P'}
\\
\\
\displaystyle\frac{
P\rarrow{\interp,V}R\quad\interp_{\min\dom\interp}(b)=1
}
{\mathbf{if}\ b\ \mathbf{then}\ P\ \mathbf{else}\ Q\rarrow{\interp,V}R}
\qquad
\frac{
Q\rarrow{\interp,V} R\quad\interp_{\min\dom\interp}(b)=0
}
{\mathbf{if}\ b\ \mathbf{then}\ P\ \mathbf{else}\ Q\rarrow{\interp,V}R}
\\
\\
\displaystyle\frac{P\rarrow{A,V}P'}{P\sqcup Q\rarrow{A,V}P'}
\qquad
\displaystyle\frac{Q\rarrow{A,V}Q'}{P\sqcup Q\rarrow{A,V}Q'}
\\
\\
\displaystyle\frac{
[\mu X.P/X]P
\rarrow{\interp,V}Q
}
{\mu X.P\rarrow{\interp,V}Q}
\\
\\
\displaystyle\frac{
\begin{array}{l}
\Var(\interp)\cup\Var(\jnterp)\subseteq\Var(\knterp)\\
\knterp|_{\Var(\interp)\cup\{r,n\}}=\interp\\
\knterp|_{\Var(\jnterp)\cup\{r,n\}}=\jnterp\\
\const^\circ(\neg r\wedge \neg n,\knterp,1)\\
V_1\cap V_2=\emptyset\\
P\rarrow{\interp,V_1} \terminate\ Q\rarrow{\jnterp,V_2} \terminate
\end{array}
}
{P\|Q\rarrow{\knterp,V_1\cup V_2} \terminate}
\qquad
\displaystyle\frac{
\begin{array}{l}
\Var(\interp)\cup\Var(\jnterp)\subseteq\Var(\knterp)\\
\knterp|_{\Var(\interp)\cup\{r,n\}}=\interp\\
\knterp|_{\Var(\jnterp)\cup\{r,n\}}=\jnterp\\
\const^\circ(\neg r\wedge \neg n,\knterp,1)\\
V_1\cap V_2=\emptyset\\
P\rarrow{\interp,V_1} P'\ Q\rarrow{\jnterp,V_2} Q'\\
P'\not=\terminate,\ Q'\not=\terminate
\end{array}
}
{P\|Q\rarrow{\knterp,V_1\cup V_2} P'\|Q'}
\\
\\
\displaystyle\frac{
\begin{array}{l}
\Var(\interp)\cup\Var(\jnterp)\subseteq\Var(\knterp)\\
\knterp|_{\Var(\interp)\cup\{r,n\}}=\interp\\
\knterp|_{\Var(\jnterp)\cup\{r,n\}}=\jnterp\\
\const^\circ(\neg r\wedge \neg n,\knterp,1)\\
V_1\cap V_2=\emptyset\\
P\rarrow{\interp,V_1} P'\ Q\rarrow{\jnterp,V_2} \terminate\\
P'\not=\terminate
\end{array}
}
{P\|Q\rarrow{\knterp,V_1\cup V_2} P'}
\qquad
\displaystyle\frac{
\begin{array}{l}
\Var(\interp)\cup\Var(\jnterp)\subseteq\Var(\knterp)\\
\knterp|_{\Var(\interp)\cup\{r,n\}}=\interp\\
\knterp|_{\Var(\jnterp)\cup\{r,n\}}=\jnterp\\
\const^\circ(\neg r\wedge \neg n,\knterp,1)\\
V_1\cap V_2=\emptyset\\
P\rarrow{\interp,V_1} \terminate\ Q\rarrow{\jnterp,V_2} Q'\\
Q'\not=\terminate
\end{array}
}
{P\|Q\rarrow{\knterp,V_1\cup V_2} Q'}
\\
\\
\displaystyle\frac{
\begin{array}{l}
V\subseteq V'\\
\const(V'\setminus V,\jnterp)\\
V'\cup\Var(\interp)\subseteq\Var(\jnterp)\\
\jnterp|_{\Var(\interp)\cup\{r,n\}}=\interp\\
\const^\circ(r\wedge n,\interp,1)\\
P\rarrow{\interp,V} P'\quad P'\not=\terminate
\end{array}
}
{P\|Q\rarrow{\jnterp,V'} P'\|Q}
\qquad
\displaystyle\frac{
\begin{array}{l}
V\subseteq V'\\
\const(V'\setminus V,\jnterp)\\
V'\cup\Var(\interp)\subseteq\Var(\jnterp)\\
\jnterp|_{\Var(\interp)\cup\{r,n\}}=\interp\\
\const^\circ(r\wedge n,\interp,1)\\
Q\rarrow{\interp,V} Q'\quad Q'\not=\terminate
\end{array}
}
{P\|Q\rarrow{\jnterp,V'} P\|Q'}
\qquad
\displaystyle\frac{
\begin{array}{l}
V\subseteq V'\\
\const(V'\setminus V,\jnterp)\\
V'\cup\Var(\interp)\subseteq\Var(\jnterp)\\
\jnterp|_{\Var(\interp)\cup\{r,n\}}=\interp\\
\const^\circ(r\wedge n,\interp,1)\\
P\rarrow{\interp,V} \terminate
\end{array}
}
{P\|Q\rarrow{\jnterp,V'} Q}
\qquad
\displaystyle\frac{
\begin{array}{l}
V\subseteq V'\\
\const(V'\setminus V,\jnterp)\\
V'\cup\Var(\interp)\subseteq\Var(\jnterp)\\
\jnterp|_{\Var(\interp)\cup\{r,n\}}=\interp\\
\const^\circ(r\wedge n,\interp,1)\\
Q\rarrow{\interp,V} \terminate
\end{array}
}
{P\|Q\rarrow{\jnterp,V'} P}
\\
\\
\end{array}$

\section{A DC semantics of Hybrid Communicating Sequential Processes}
\label{adcsemanticsforHCSP}

Given a process $P$, $\sem{P}$, with some subscripts to be specified below, is a DC formula which defines the class of DC interpretations that represent runs of $P$.

\paragraph{\bf Process variables and their corresponding DC temporal variables}
Real-valued process variables $x$ are modelled by a pairs of DC temporal variables $x$ and $x'$, which are meant to store the value of $x$ at the beginning and at the end of the reference interval, respectively. The axiom
\[\Box\forall z\neg(x'=z\chop x\not=z).\]
entails that the values of $x$ and $x'$ are determined by the beginning and the end point of the reference interval, respectively. It can be shown that
\[\begin{array}{l}
\models_{DC}\Box\forall z\neg(x'=z\chop x\not=z)\Rightarrow \\
\quad \Box(x=c\Rightarrow\neg(x\not=c\chop \top))\wedge\Box(x'=c\Rightarrow\neg(\top\chop x'\not=c))
\end{array}\]
This is known as the locality principle in ITL about $x$. About primed variables $x'$, the locality principle holds wrt the endpoints of reference intervals. Boolean process variables are similarly modelled by propositional temporal letters. For the sake of brevity we put
\[
 \loc(X)\rightleftharpoons\bigwedge\limits_{x\in X\atop x\ \mbox{\tiny is real} }\Box\forall z\neg(x'=z\chop x\not=z)\wedge \bigwedge\limits_{x\in X\atop x\ \mbox{\tiny is boolean} }\Box\neg((x'\chop \neg x)\vee(\neg x'\chop x))
 \]
In the sequel, given a DC term $e$ or formula $\varphi$ written using only unprimed variables, $e'$ and $\varphi'$ stand for the result of replacing all these variables by their respective primed counterparts.

\paragraph{Time derivatives of process variables}
As mentioned above, terms of the form $\dot x$ where $x$ is a process variable are treated as
distinct process variables and are modelled by their respective temporal variables $\dot x$ and $\dot x'$. The requirement on $\dot x$ to be interpreted as the time derivative of $x$ is incorporated in the semantics of continuous evolution statements.

\paragraph{Computation time and the parameters $\sem{.}$}
As explained in Section \ref{opsem}, we allow stretches of time that are dedicated to computation steps and are marked by the auxiliary boolean process variable $r$. Such stretches of time are conveniently excluded when calculating the duration of process execution. To this end, in DC formulas, we use a state variable $R$ which indicates the time taken by computation steps by the reference process. Similarly, a state variable $N$ indicates time for computation steps by which any process that runs in parallel with the reference one, including the reference one. $R$ and $N$ match the auxiliary variables $r$ and $n$ from the operational semantics and, just like $r$ and $n$, are supposed to satisfy the condition $R\Rightarrow N$. We assume that all continuous evolution becomes temporarily suspended during intervals in which computation is performed, with the relevant real quantities and their derivatives remaining frozen. To guarantee the atomicity of assignment, computation intervals of different processes are not allowed to overlap. As it becomes clear in the DC semantics of $\|$ below, $P_i$ of $P_1\|P_2$ are each given its own variable $R_i$, $i=1,2$, to mark computation time, and $R_1$ and $R_2$ are required to satisfy the constraints $\neg(R_1\wedge R_2)$ and $R_1\vee R_2\Leftrightarrow R$ where $R$ is the variable which marks computation times for the whole of $P_1\|P_2$.

The semantics $\sem{P}_{R,N,V}$ of a HCSP term $P$ is given in terms of the DC temporal variables which correspond to the process variables occurring in $P$, the state variables $R$ and $N$, and the set of variables $V$ which are controlled by $P$'s immediately enveloping $\|$-operand.

\paragraph{Assignment} %znj
To express that the process variables from $X\subseteq V$ may change at the end of the reference interval only, and those from $V\setminus X$ remain unchanged, we write
{\small \[
\begin{array}{ll}
\const(V,X)\rightleftharpoons&
\bigwedge\limits_{x\in V\setminus X\atop x\ \mbox{\tiny is real}}\Box(x'= x)\wedge\bigwedge\limits_{x\in V\setminus X\atop x\ \mbox{\tiny is boolean}}\Box(x\Leftrightarrow x')\wedge\\
&
\bigwedge\limits_{x\in X\atop x\ \mbox{\tiny is real}}\Boxdot(x'= x)\wedge\bigwedge\limits_{x\in X\atop\ x\ \mbox{\tiny is boolean}}\Boxdot(x'\Leftrightarrow x).
\end{array}\]}
The meaning of simultaneous assignment is as follows:
\[\begin{array}{ll}\label{asssem}
  \textstyle
\sem{x_1,\ldots,x_n:=e_1,\ldots,e_n}_{R,N,V}\rightleftharpoons&
\lceil R\rceil_\fin\wedge \const(V,\{x_1,\ldots,x_n\})\wedge\\
&\bigwedge\limits_{i=1,\ldots,n\atop x_i\ \mbox{\tiny is real}}x_i'=e_i\wedge\bigwedge\limits_{i=1,\ldots,n\atop x_i\ \mbox{\tiny is boolean}}x_i'\Leftrightarrow e_i.
\end{array}
\]
%\mycomment{I didn't find the definition of $\lceil R\rceil_\fin$.}
%dpg it is there now.

\paragraph{Parallel composition}
Consider processes $P_1$ and $P_2$ and $V\supseteq \VarA(P_1\|P_2)$. Let $\overline{1}=2$, $\overline{2}=1$.
Let
\[{\small
\Esplit(R,R_1,R_2,V,P_1,P_2)\varphi\rightleftharpoons
\exists R_1\exists R_2
\left(\begin{array}{l}
\lceil (R_1\vee R_2\Leftrightarrow R)\wedge\neg(R_1\wedge R_2)\rceil^0\wedge\\
\bigwedge\limits_{i=1}^2\Box(\lceil R_i\rceil\Rightarrow\const(V\setminus \VarA(P_i)))
\wedge
\varphi
\end{array}\right).}\]
$\Esplit(R,R_1,R_2,V_1,V_2)$ means that
\begin{itemize}
\item the $R$-subintervals for the computation steps of $P_1\|P_2$ can be divided into $R_1$- and $R_2$-subintervals to mark the computation steps of some sub-processes $P_1$ and $P_2$ of $P$ which run in parallel;
\item the variables which are not controlled by $P_i$ remain unchanged during $P_i$'s computation steps,  $i=1,2$, and, finally,
\item some property $\varphi$, which can involve $R_1$ and $R_2$, holds.
\end{itemize}
The universal dual $\Asplit$ of $\Esplit$ is defined in the usual way.
Let $V_i$ abbreviate $\VarA(P_i)$. Now we can define $\sem{P_1\|P_2}_{R,N,V}$ as
{\small\begin{equation}\label{parsem}
\Esplit(R,R_1,R_2,V,P_1,P_2)
\bigvee\limits_{i=1}^2\left(\begin{array}{l}
\sem{P_i}_{R_i,N,V_i}\wedge(\lceil N\wedge \neg R_{\overline{i}}\rceil^0_\fin \chop \sem{P_{\overline{i}}}_{R_{\overline{i}},N,V_{\overline{i}}}\chop \lceil\neg R_{\overline{i}}\rceil^0)\vee\\
(\lceil N\wedge \neg R_i\rceil^0_\fin \chop \sem{P_i}_{R_i,N,V_i})\wedge(\sem{P_{\overline{i}}}_{R_{\overline{i}},N,V_{\overline{i}}}\chop \lceil \neg R_{\overline{i}}\rceil^0)
\end{array}\right)
\end{equation}}
To understand the four disjunctive members of $\Phi$ above, note that $P_1\|P_2$ always starts with some action on behalf of either $P_1$, or $P_2$, or both, in the case of continuous evolution. Hence (at most) one of $P_i$, $i=1,2$, needs to allow negligible time for $P_{\overline{i}}$'s first step. This is expressed by a $\lceil N\wedge\neg R_i\rceil^0_\fin$ before $\sem{P_i}_{R_i,N,V_i}$. The amount of time allowed is finite and may be $0$ in case both $P_1$ and $P_2$ start with continuous evolution in parallel. This makes it necessary to consider two cases, depending on which process starts first. If $P_{\overline{i}}$ terminates before $P_i$, then a $\lceil N\wedge\neg R_i\rceil^0$ interval follows $\sem{P_{\overline{i}}}_{R_{\overline{i}},N,V_{\overline{i}}}$. This generates two more cases to consider, depending on the value of $i$.

The definitions of $\sem{x_1,\ldots,x_n:=e_1,\ldots,e_n}_{R,N,V}$ and $\sem{P_1\|P_2}_{R,N,V}$ already appear in (\ref{asssem}) and (\ref{parsem}). Here follow the definitions for the rest of the basic constructs:

\qquad

\noindent
{\small
$\begin{array}{lll}
\oomit{\sem{\terminate}_{R,N,V} & \rightleftharpoons & \const(V)\wedge\lceil\neg R\rceil^0\\
[5mm]
\sem{\mathbf{skip}}_{R,N,V} & \rightleftharpoons &
\ell=0\\
[5mm]}
\sem{\mathbf{await}\ b}_{R,N,V} & \rightleftharpoons & \const(V)\wedge(\lceil\neg R \rceil\vee\ell=0)\wedge\Boxdot\neg \overline{b'}\wedge(\overline{b'}\vee\ell=\infty)\\
\sem{\langle F(\dot x,x)=0\wedge b\rangle}_{R,N,V} & \rightleftharpoons &
 \left(\begin{array}{l}
\const(V\setminus\{\dot x,x\})\wedge\lceil\neg R\rceil\wedge\\
\Box\left(\begin{array}{l}
\lceil N\rceil\Rightarrow\const(\{\dot x,x\})\wedge\\
\forall\mathit{ub}\Box(\dot x\leq\mathit{ub})\Rightarrow x'\leq x+\mathit{ub}\int \!\!\neg N\wedge\\
\forall\mathit{lb}\Box(\dot x\geq\mathit{lb})\Rightarrow x'\geq x+\mathit{lb}\int \!\!\neg N\wedge\\
F(\dot x,x)=0
\end{array}\right)\wedge
\BoxSdot b
\end{array}\right)\\
&&\chop (\neg b\wedge\ell=0) \\
[5mm]
%dpg added the clause for sequential composition
\sem{P;Q}_{R,N,V} & \rightleftharpoons & (\sem{P}_{R,N,V}\chop \lceil N \wedge
\neg R\rceil^0_\fin\chop \sem{Q}_{R,N,V})\\
\sem{P\sqcup Q}_{R,N,V} & \rightleftharpoons & \sem{P}_{R,N,V}\vee\sem{Q}_{R,N,V}\\
[5mm]
\sem{\mathbf{if}\ b\ \mathbf{then}\ P\ \mathbf{else}\ Q}_{R,N,V} & \rightleftharpoons & (b\wedge \sem{P}_{R,N,V})\vee (\neg b\wedge \sem{Q}_{R,N,V})\\
[5mm]
\sem{\mathbf{while}\ b\ \mathbf{do}\ P}_{R,N,V} & \rightleftharpoons & (b\wedge\sem{P}_{R,N,V}\chop \lceil\neg R\rceil^0)^*\chop(\neg b\wedge\ell=0)\\
\end{array}$
}

\noindent
To understand $\sem{\langle F(\dot x,x)=0\wedge b\rangle}_{R,N,V}$, observe that, assuming $I$ to be the DC interpretation in question and $\lambda t.I_t(\dot x)$ to be continuous, the two inequalities in
$\sem{\langle F(\dot x,x)=0\wedge b\rangle}_{R,N,V}$
express that
\[\textstyle I_{t_2}(x)-I_{t_1}(x)=\int\limits_{t_1}^{t_2} I_t(\dot x)(1-I_t(N))dt\]
at all finite subintervals $[t_1,t_2]$ of the reference intervals. This means that both $\dot x$ and $x$ are constant in $N$-subintervals, and $I_{t_2}(x)-I_{t_1}(x)=\int\limits_{t_1}^{t_2}I_t(\dot x)dt$ holds at $\neg N$-subintervals.

\paragraph{Completeness of $\sem{.}$}
Given a process term $P$, every DC interpretation $I$ for the vocabulary of $\sem{P}_N,N,V$ represents a valid behaviour of $P$ with $N$ being true in the subintervals which $P$ uses for computation steps. To realize this, consider HCSP terms  $P$, $Q$ of the syntax
\begin{equation}\label{gnf}
\begin{array}{lll}
P, Q &::= & {\bf skip}  \mid A; R\mid []_i A_i;R_i \mid  \mathbf{if}\ b\ \mathbf{then}\ P\ \mathbf{else}\ Q \mid   P \sqcup Q\\
A &::= & x :=e\mid{\bf await}\ b \mid
\langle F(\dot x, x)=0\wedge b\rangle
\end{array}
\end{equation}
where $R$ and $R_i$ stand for a process term with no restrictions on its syntax (e.g., occurrences of $\mathbf{while}$-terms are allowed). (\ref{gnf}) is the {\em guarded normal form (GNF)} for HCSP terms, with the guards being the primitive terms of the form $A$, and can be established by induction on the construction of terms, with suitable equivalences for each combination of guarded operands that $\|$ may happen to have. E.g.,
\begin{equation}\label{examplegnf0}
\langle F_1(\dot x, x)=0\wedge b_1\rangle;P_1\|\langle F_2(\dot x, x)=0\wedge b_2\rangle;P_2
\end{equation}
is equivalent to
\begin{equation}\label{examplegnf}
\begin{array}{ll}
\langle F_1(\dot x, x)=0\wedge F_2(\dot x, x)=0\wedge b_1\wedge b_2\rangle;\\
 \mathbf{if}\ b_1\ \mathbf{then}\ \langle F_1(\dot x, x)=0\wedge b_1\rangle;P_1\|P_2\\
 \mathbf{else}\ \mathbf{if}\ b_2\ P_1\|\langle F_2(\dot x, x)=0\wedge b_2\rangle;P_2\\
 \mathbf{else}\ P_1\|P_2
\end{array}
\end{equation}
Note that $[]_i A_i;R_i$ is a modest generalization of $IO$ as defined in (\ref{IOsyntax}). Some combinations of operands of $\|$ require external choice to be extended this way, with the intended meaning being that if none of the $A_i$s which have the forms $ch?x$ and $ch!e$ is ready to execute, then some other option can be pursued immediately. For example, driving $\|$ inwards may require using that
\[\begin{array}{l}
(ch_1?x\rightarrow P_1[] ch_2!e\rightarrow P_2 [] ch_3?y\rightarrow P_3) \| ch_1! f;Q_1\| ch_2? z;Q_2\ \equiv\\
\qquad ((x:=f;(P_1\|Q_1)\|ch_2? z;Q_2) \sqcup (z:=e;(P_2\| Q_2)\|ch_1! f;Q_1))[] \\
\qquad ch_3? y;P_3\| ch_1! f;Q_1\| ch_2? z;Q_2\ .
\end{array}\]
On the RHS of $\equiv$ above, one of the assignments and the respective subsequent process are bound to take place immediately in case the environment is not ready to communicate over $ch_3$.

The GNF renders the correspondence between the semantics of guards and the $\interp$s which appear in the rules for $\rarrow{A,V}$ explicit, thus making obvious that any finite prefix of a valid behaviour satisfies some {\em chop}-sequence of guards that can be generated by repeatedly applying the GNF a corresponding number of times and then using the distributivity of {\em chop} over disjunction, starting with the given process term, and then proceeding to transform the $R$-parts of guarded normal forms that appear in the process. The converse holds too. This entails that the denotational semantics is equivalent to the operational one.

\section{Reasoning about Hybrid Communicating Sequential Processes with DC Hoare triples}
\label{hoaresystem}

We propose reasoning in terms of triples of the form
\begin{equation}\label{htriple}
\{A\}P\{G\}_V
\end{equation}
where $A$ and $G$ are DC formulas, $P$ is a HCSP term, and $V$ is a set of program variables. $V$ is supposed to denote the variables whose evolution needs to be specified in the semantics of $P$, e.g., an assignment $x:=e$ in $P$ is supposed to leave the values of the variables $y\not=x$ unchanged. This enables deriving, e.g., $\{y=0\}x:=1\{y=0\}_{\{x,y\}}$, which would be untrue for a $y$ that belongs to a process that runs in parallel with the considered one and is therefore not a member of $V$. Triple (\ref{htriple}) is valid, if
\begin{equation}\label{triplesem}
\textstyle
\models
\loc(V)\wedge\lceil R\Rightarrow N\rceil^0
\wedge (A\chop \top)\Rightarrow\neg(\sem{P}_{R,N,V}\wedge\neg G\chop \top)
\end{equation}
Since $R$ and $N$ typically occur in $\sem{P}_{R,N,V}$, triples (\ref{htriple}) can have occurrences of $R$ and $N$ in $A$ and $G$ too, with their intended meanings.

Next we propose axioms and rules for deriving triples about processes with each of the HCSP constructs as the main one in them. For $P$ of one of the forms $\mathbf{skip}$,
$x_1,\ldots,x_n:=e_1,\ldots,e_n$, and $\langle{\mathcal F}(\dot x, x)=0\wedge b\rangle$, we introduce the axioms
\[\{\top\}P\{\sem{P}_{R,N,V}\}_V.\]
where $V$ can be any set of process variables such that $V\supseteq \VarA(P)$.
Here follow the rules for reasoning about processes which are built using each of the remaining basic constructs:
\[{\small\begin{array}{ll}
(\mathbf{seq})&
% mu-seq
\oomit{
\displaystyle
\frac{
\begin{array}{l}
\{A\}P\{G\}_V
\\
\{B\}Q\{H\}_V\\
\loc(V)\wedge\lceil R\Rightarrow N\rceil^0\wedge\Diamond_l^c A\Rightarrow\neg (G\wedge\ell<\infty\chop \neg\Diamond_l^c(\lceil N\wedge\neg R\rceil^0 \chop B))
\end{array}
}{
\{A\}P;Q\{(G\chop \lceil N\wedge\neg R\rceil^0 \chop H)\}_V
}
}%oomit mu sec
\displaystyle
\frac{
\begin{array}{l}
\{A\}P\{G\}_V\qquad \{B\}Q\{H\}_V\\
\loc(V)\wedge\lceil R\Rightarrow N\rceil^0\wedge (A\chop \top)\Rightarrow\neg(G\chop \neg(\lceil N\wedge\neg R\rceil^0\chop B\chop \top))
\end{array}}{
\{A\}P;Q\{(G\chop \lceil N\wedge\neg R\rceil^0 \chop H)\}_V
}
\\
[3mm]
(\sqcup)&
\displaystyle
\frac{
\{A\}P\{G\}_V\qquad\{B\}Q\{H\}_V
}{
\{A\wedge B\}P\sqcup Q\{G\vee H\}_V
}\\
[3mm]
(\mathbf{if})&
\displaystyle
\frac{
\{A\wedge b\}P\{G\}_V\qquad\{A\wedge\neg b\}Q\{G\}_V
}{
\{A\}\mathbf{if}\ b\ \mathbf{then}\ P\ \mathbf{else}\ Q\{G\}_V
}\\
[3mm]
(\mathbf{while})&
% mu-while
\oomit{
\displaystyle
\frac{
\{A\}P\{G\}_V\qquad
\loc(V)\wedge\lceil R\Rightarrow N\rceil^0\wedge (A\chop \top)\Rightarrow\neg(G\chop \neg(\lceil N\wedge\neg R\rceil^0\chop A\chop \top))
}{
\{A\}\mathbf{while}\ b\ \mathbf{do}\ P\left\{\begin{array}{l}
(\ell = 0\wedge\neg b)\vee\\
((b\wedge G\chop \lceil\neg R\rceil^0 )^*\chop b\wedge G\chop \neg b\wedge\ell=0)
\end{array}\right\}_V
}
}
%oomit mu-while
\displaystyle{
\frac{
\begin{array}{l}
   \{A\}P\{G\}_V\\
   \loc(V)\wedge\lceil R\Rightarrow N\rceil^0\wedge(A\chop \top) \Rightarrow\neg(G\wedge\ell<\infty\chop \neg\Diamond_l^c(\lceil\neg R\rceil^0\chop A))
\end{array}
}{
\{A\}\mathbf{while}\ b\ \mathbf{do}\ P\{((b\wedge G\chop \lceil\neg R\rceil^0)^*\chop \neg b\wedge\ell=0)
\}_V
}
}
\\
\end{array}}
\]

\paragraph{Parallel composition}
The established pattern suggests the following proof rule $(\|)$:

\[{\small\begin{array}{ll}
\displaystyle
\frac{
\{A_1\}P_1\{G_1\}_{\VarA(P_1)}\qquad\{A_2\}P_2\{G_2\}_{\VarA(P_2)}
}{
\begin{array}{l}
\left\{\Asplit(R,R_1,R_2,V,P_1,P_2)\left(
\bigvee\limits_{i=1}^2\neg(\lceil N\wedge\neg R_i\rceil^0_\fin \chop \neg [R_i/R]A_i\chop \top)\wedge ([R_{\overline{i}}/R]A_{\overline{i}}\chop \top)
\right)
\right\}\\
P_1\|P_2\\
\left\{\Esplit(R,R_1,R_2,V,P_1,P_2)\left(
\bigvee\limits_{i=1}^2
\begin{array}{l}
G_i\wedge(\lceil N\wedge\neg R_{\overline{i}}\rceil^0_\fin\chop G_{\overline{i}}\chop \lceil \neg R_{\overline{i}}\rceil^0)\\
\vee\\
(\lceil N\wedge \neg R_i\rceil^0_\fin\chop G_i)\wedge(G_{\overline{i}}\chop \lceil \neg R_{\overline{i}}\rceil^0)
\end{array}\right)
\right\}_V
\end{array}
}
\end{array}}\]

This rule can be shown to be complete as it straightforwardly encodes the semantics of $\|$. However, it is not convenient for proof search as it only derives triples with a special form of the condition on the righthand side and actually induces the use of $\sem{P_i}_{R_i,N,\VarA(P_i)}$ as $G_i$, which typically give excess detail. We discuss a way around this inconvenience below, together with the modifications of the setting which are needed in order to handle general HCSP fixpoints $\mu X.P$.
\paragraph{General rules}
Along with the process-specific rules, we also need the rules
\[(N)\qquad\frac{\loc(V)\wedge\lceil R\Rightarrow N\rceil^0\wedge\Diamond_l^c A\Rightarrow  G}{\{A\}P\{G\}_V}\qquad\VarA(P)\subset V\]
\[({\bf K})\qquad\frac{
\{A\}P\{G\Rightarrow H\}_V \quad \{B\}P\{G\}_V \quad \loc(V)\wedge\lceil R\Rightarrow N\rceil^0\wedge\Diamond_l^c A\Rightarrow\Diamond_l^c B}{\{A\}P\{H\}_V}\]
These rules are analogous to the modal form $N$ of G\"odel's generalization rule (also known as the necessitation rule) and the modal axiom ${\bf K}$.

\subsubsection*{Soundness and Completeness}

The soundness of the proof rules is established by a straightforward induction on the construction of proofs with the definition of $\sem{.}_{R,N,V}$. The system is also complete relative to validity in DC. This effectively means that we allow all valid DC formulas as axioms in proofs, or, equivalently, given some sufficiently powerful set of proof rules and axioms for the inference of valid {\em DC formulas} in DC with infinite intervals, our proof rules about triples suffice for deriving all the {\em valid triples}. Such systems are beyond the scope of our work. A Hilbert-style proof system for ITL with infinite intervals was proposed and shown to be complete with respect to an abstractly defined class of time models (linearly ordered commutative groups) in \cite{WX04}, building on similar work about finite intervals from \cite{Dut95}.

The deductive completeness of our proof system boils down to the possibility to infer triples of the form $\{\top\}P\{G\}_V$ for any given term $P$ and a certain strongest corresponding $G$, which turns out to be the DC formula $\sem{P}_{R,N,V}$ that we call the semantics of $P$. Then the validity of $\{\top\}P\{\sem{P}_{R,N,V}\}_V$ is used to derive any valid triple about $P$ by a simple use of the proof rules ${\bf K}$ and $N$. The completeness now follows from the fact that $\sem{P}_{R,N,V}$ defines the class of all valid behaviours of $P$.

\begin{proposition}\label{prop1}
The triple
\begin{equation}\label{semtriple}
\{\top\}P\{\sem{P}_{R,N,V}\}_V
\end{equation}
is derivable for all process terms $P$ and all $V$ such that $\VarA(P)\supseteq V$.
\end{proposition}
\begin{proof}
Induction on the construction of $P$. The triple (\ref{semtriple}) is an axiom for $P$ being the forms $\mathbf{skip}$,
$x_1,\ldots,x_n:=e_1,\ldots,e_n$, and $\langle{\mathcal F}(\dot x, x)=0\wedge b\rangle$. For processes $P$ of other forms, the induction step follows by single applications of the corresponding rules to the instances of (\ref{semtriple}) which are assumed to hold for $P$'s subprocesses.
\end{proof}

\begin{corollary}[relative completeness of the Hoare-style proof system]
Let $A$, $G$ and $P$ be such that (\ref{triplesem}) is valid. Then (\ref{htriple}) is derivable in the extension of the given proof system by all DC theorems.
\end{corollary}
\begin{proof}
By $N$ we first derive
\[\{A\}P\{\sem{P}_{R,N,V}\Rightarrow G\}_V.\]
Now, using (\ref{semtriple}) from Proposition \ref{prop1}, the validity of $\Diamond_l^c A\Rightarrow\Diamond_l^c\top$ in DC and ${\bf K}$, we derive (\ref{htriple}).
\end{proof}

%\section{An assume-guarantee proof rule for $\|$}

\section{General fixpoints and bottom-up proof search in HCSP}
\label{muSection}

To avoid the constraints on the form of the conclusion of rule $(\|)$, we propose a set of rules which correspond to the various possible forms of the operands of the designated $\|$ in the considered HCSP term. These rules enable bottom-up proof search much like when using the rules for (just) CSP constructs, which is key to the applicability of classical Hoare-style proof. We propose separate rules for each combination of main connectives in the operands of $\|$, except $\|$ itself and the fixpoint construct. For instance, the equivalence between (\ref{examplegnf0}) and (\ref{examplegnf}) suggests the following rule for this particular combination of $\|$ with the other connectives:
\[\frac{\begin{array}{l}
\{P\}\langle F_1(\dot x, x)=0\wedge F_2(\dot x, x)=0\wedge b_1\wedge b_2\rangle\{R\}\\
\{R\wedge b_1\wedge\neg b_2\}\langle F_1(\dot x, x)=0\wedge b_1\rangle;P_1\|P_2\{Q\}\\
\{R\wedge b_2\wedge\neg b_1\}P_1\|\langle F_2(\dot x, x)=0\wedge b_2\rangle;P_2\{Q\}\\
\{R\wedge\neg b_1\wedge\neg b_2\}P_1\|P_2\{Q\}
\end{array}}
{\{P\}\langle F_1(\dot x, x)=0\wedge b_1\rangle;P_1\|\langle F_2(\dot x, x)=0\wedge b_2\rangle;P_2\{Q\}}\]
Rules like the above one use the possibility to drive $\|$ inwards by equivalences like that between (\ref{examplegnf0}) and (\ref{examplegnf}), which can be derived for all combinations of the main connectives of $\|$'s operands, except for loops, and indeed can be used to eliminate $\|$. For $\mathbf{while}$-loops, the GNF contains a copy of the loop on the RHS of {\em chop}:
\begin{equation}\label{whilegnf}\mathbf{while}\ b\ \mathbf{do}\ P\equiv\mathbf{if}\ b\ \mathbf{then}\ (P;\mathbf{while}\ b\ \mathbf{do}\ P)\ \mathbf{else}\ \mathbf{skip}.
\end{equation}
Tail-recursive instances of $\mu X.G$ come handy in completing the elimination of $\|$ in such cases by standard means, namely, by treating equivalences such as (\ref{whilegnf}) as the equations leading to definitions such as (\ref{whiledef}).

To handle general recursion in our setting, we need to take care of the special way in which we treat $A$ from $\{A\}P\{G\}$ in (\ref{triplesem}). In the rule for $\{A\}\mathbf{while}\ b\ \mathbf{do}\ P\{G\}$ clipping of initial $G$-subintervals of an $A$-interval are supposed to leave us with suffix subintervals which satisfy $(A\chop \top)$, to provide for successive executions of $P$. With $X$ allowed on the LHS of {\em chop} in the $P$ of $\mu X.P$, special care needs to be taken for this to be guaranteed. To this end, instead of $(A\chop \top)$, $\Diamond_l^c A$ is used to state that $A$ holds at an interval that starts at the same time point as the reference one, and is not necessarily its subinterval. This is needed for reasoning from the viewpoint of intervals which accommodate nested recursive executions. The meaning of triples (\ref{htriple}) becomes
\begin{equation}\label{triplesemmu}
\textstyle
\models
\loc(V)\wedge\lceil R\Rightarrow N\rceil^0
\wedge \Diamond_l^c A\wedge \sem{P}_{R,N,V}\Rightarrow G.
\end{equation}
In this setting, $\mu X.P$ admits the proof rule, where $X\mbox{ does not occur in }A$:
\[(\mu)\qquad
\displaystyle
\frac{\loc(V)\wedge\lceil R\Rightarrow N\rceil^0\wedge\Diamond_l^c A\wedge G \Rightarrow[\Diamond_l^c A\wedge X/X]G
\qquad\{A\}P\{G\}_V}{\{A\}\mu X.P\{\mu X.G\}_V}
%\qquad X\mbox{ does not occur in }A.
\]
This rule subsumes the one for $\mathbf{while}-\mathbf{do}$, but only as part of a suitably revised variant of the whole proof system wrt (\ref{triplesemmu}). E.g., the rule for sequential composition becomes
\[\displaystyle
\frac{
\begin{array}{l}
\{A\}P\{G\}_V
\\
\{B\}Q\{H\}_V\\
\loc(V)\wedge\lceil R\Rightarrow N\rceil^0\wedge\Diamond_l^c A\Rightarrow\neg (G\wedge\ell<\infty\chop \neg\Diamond_l^c(\lceil N\wedge\neg R\rceil^0 \chop B))
\end{array}
}{
\{A\}P;Q\{(G\chop \lceil N\wedge\neg R\rceil^0 \chop H)\}_V
}\]

\section{Related work}

Our work was influenced by the studies on DC-based reasoning about process-algebraic specification languages in \cite{ZH00,DBLP:journals/fac/HeX03,DBLP:conf/fm/YongG99,DBLP:conf/ifm/HaxthausenY00}. In a previous paper we proposed a calculus for HCSP \cite{DBLP:conf/aplas/LiuLQZZZZ10}, which was based on DC in a limited way and lacked compositionality.  In \cite{WZG12} and \cite{GWZZ13} we gave other variants of compositional and sound calculi for HCSP with different assertion and temporal condition formats. Completeness was not considered. The approach in \cite{GWZZ13} and in this paper is largely drawn from \cite{GD02} where computation time was proposed to be treated as negligible in order to simplify delay calculation, and the operator of projection was proposed to facilitate writing requirements with negligible time ignored. Hoare-style reasoning about real-time systems was also studied in the literature with explicit time logical languages \cite{DBLP:journals/fac/Hooman94}. However, our view is that using temporal logic languages is preferable. Dedicated temporal constructs both lead to more readable specifications, and facilitate the identification of classes of requirements that can be subjected to automated analysis.  Another approach to the verification of hybrid systems is Platzer's Differential Dynamic Logic \cite{DBLP:journals/jar/Platzer08}. However, the hybrid programs considered there have limited functionality. Communication, parallelism and interrupts are not handled. For logic compositionality, assume-guarantee reasoning has been studied for communication-based concurrency in CSP without timing in \cite{DBLP:journals/tse/MisraC81,DBLP:journals/dc/PandyaJ91}.

Both in our work and in alternative approaches such as \cite{DBLP:journals/tse/MisraC81}, the treatment of continuous evolution is somewhat separated from the analysis of the other basic process-algebraic constructs. Indeed we make a small step forward here by fully expressing the meaning of differential law-governed evolution in DC, which is theoretically sufficient to carry out all the relevant reasoning in the logic. Of course, the feasibility of such an approach is nowhere close to the state of art in the classical theory of ordinary differential equations. Indeed it would typically lead to formalized accounts of classical reasoning. Techniques for reasoning about the ODE-related requirements are the topic of separate studies, see, e.g., \cite{DBLP:conf/hybrid/PrajnaJ04,DBLP:conf/hybrid/SankaranarayananSM04,DBLP:conf/emsoft/LiuZZ11}.

\section*{Concluding remarks}

We have presented a weakly monotonic time-based semantics and a corresponding Hoare style proof system for HCSP with both the semantics and the temporal conditions in triples being in first-order DC with infinite intervals and extreme fixpoints. The proof system is compositional but the proof rule for parallel composition introduces complications because of the special form of the triples that it derives. However, we have shown that HCSP equivalences that can serve as elimination rules for $\|$ can also be used to derive proof rules for $\|$ which do not bring the above difficulty and indeed are perfectly compatible with standard bottom-up proof search. Interestingly, the informal reading of the derived rules for $\|$ together with the ones which are inherited from CSP, does not require the mention of weakly monotonic time technicalities. This means that the use of this special semantics can be restricted to establishing the soundness of practically relevant proof systems and awareness of its intricacies is not essential for applying the system. The meaning of triples we propose subsumes classical pre-/postcondition Hoare triples and triples linking (hybrid) temporal conditions in a streamlined way. This is a corollary of the choice to use assumptions which hold at an arbitrary initial subintervals, which is also compatible with reasoning about invariants $A$ in terms of statements of the form $(A\chop \top)\Rightarrow\neg(\sem{P}\wedge \neg(A\chop \top))$.

\bibliographystyle{plain}
\bibliography{rg}

\end{document}